\documentclass[12pt]{amsart}
\usepackage{amsfonts}
\usepackage{layout}
\usepackage[top=2cm, bottom=1.5cm, left=2cm, right=2cm]{geometry}
\usepackage{amssymb}
\usepackage{amsmath}
\usepackage{amsthm, mathrsfs}
\usepackage{hyperref}

\hypersetup{colorlinks=true,linkcolor=red,urlcolor=red,citecolor=blue}
\makeatletter

\@addtoreset{equation}{section}
\makeatother
\newtheorem{Thm}{Theorem}[section]

\newtheorem{propo}{Proposition}[section]

\newtheorem{remark}{Remark}[section]

\begin{document}
\title{Polyanalytic Reproducing Kernels on the Quantized Annulus}
\author[N. Demni]{Nizar Demni}
\address{IRMAR, Universit\'e de Rennes 1\\
Campus de Beaulieu\\
35042 Rennes cedex\\
France}
\author[Z. Mouayn]{Zouhair Mouayn}
\address{Department of Mathematics\\
Faculty of Sciences and Technics (M'Ghila)\\
Sultan Moulay Slimane \\
PO. Box 523, B\'eni Mellal\\
Morocco}
\maketitle

\begin{abstract}
While dealing with the constant-strength magnetic Laplacian on the annulus,
we complete J. Peetre's work. In particular, the eigenspaces associated with
its discrete spectrum are true-polyanalytic spaces with respect to the
invariant Cauchy-Riemann operator, and we write down explicit formulas for
their reproducing kernels. The latter are expressed by means of the fourth
Jacobi theta function and of its logarithmic derivatives when the magnetic
field strength is an integer. Under this quantization condition, we also
derive the transformation rule satisfied by the reproducing kernel under the
automorphism group of the annulus.
\end{abstract}

\section{Introduction}

\textbf{\ }A Riemann surface $\mathcal{M}$ is called hyperbolic if its
holomorphic universal covering space is the unit disk $\mathbb{D=}$ $\left\{
\zeta \in \mathbb{C},\left\vert \zeta \right\vert <1\right\} $. If
additionally its fundamental group $\pi _{1}\left( \mathcal{M}\right) $ is
commutative then $\mathcal{M}$ is isomorphic to either $\mathbb{D}$, or to the
punctured unit disk $\mathbb{D\setminus }\left\{ 0\right\} $ or to the
annulus 
\begin{equation}
\Omega _{1,R}:=\left\{ z\in \mathbb{C},1<\left\vert z\right\vert <R\right\} .
\tag{1.1}
\end{equation}%
Such surfaces are referred to as exceptional hyperbolic Riemann surfaces 
\cite{Far} and are involved in physics as phase spaces for the classical mechanics of systems with Hamiltonian
functions. A special interest is given to the annulus $\Omega _{1,R}$ and stems from the key role it plays in the quantization of the Hall effect.
Indeed, because of its complete universality, the quantization must be
insensitive to continuous deformations of the sample geometry. Using this
freedom degree and due to its additional symmetry, the annulus geometry was
proposed by Laughlin \cite{Laug} as a substitute of the standard "Hall bar"
one.

Recently, the annular domain appeared in the study of the confinement of exciton-polariton condensate in relation to the Meissner effect \cite{DDAI}.
In this respect, topological spin Meissner states can be observed at arbitrary high magnetic fields. For more details on this phenomenon, we refer
the reader to (\cite{Taka}, p.154) where a magnetic flux quantization is also discussed. The annulus appeared as well in relation to the problem of a Josephson junction in a superconducting loop subject to a uniform
magnetic field (see \cite{Badia} and references therein). In a nutshell, the occurrence of the annular geometry becomes has increased in both theoretical and experimental physics.

Geometrically, the annulus $\Omega _{1,R}$ can be covered by a horizontal
strip using the exponential map and the density of its corresponding Poincar\'{e} metric reads (see e.g. \cite{Pe,Pee}): 
\begin{equation}
\omega _{R}\left( z\right) :=\left( \frac{\log R}{\pi }\right) \left\vert
z\right\vert \sin \left( \frac{\pi \log \left\vert z\right\vert }{\log R}%
\right) .  \label{Poincare}
\end{equation}%
We can therefore consider for any $B>1/2$ the associated weighted $L^{2}$%
-space $\mathfrak{H}_{B}\left( \Omega _{1,R}\right) $ of functions $\phi
:\Omega _{1,R}\rightarrow \mathbb{C}$ with finite squared norm: 
\begin{equation}
\quad \int_{\Omega _{1,R}}|\phi (z)|^{2}\left( \omega _{R}\left( z\right)
\right) ^{2B-2}d\mu \left( z\right) <\infty ,  \tag{1.3}
\end{equation}%
$d\mu \left( z\right) $ being the Lebesgue measure on $\mathbb{C=R}^{2}.$
The holomorphic subspace $\mathcal{A}\left( \Omega _{1,R}\right) $ of $%
\mathfrak{H}_{B}\left( \Omega _{1,R}\right) $ was considered by Peetre \cite%
{Pee} where the correspondence principle \cite{Berezin} in the
semi-classical limit $B\rightarrow +\infty $ was proved with $B$ playing the role of the
inverse Planck constant $\hslash $. This principle was connected with
the holomorphic Berezin transform whose integral kernel is expressible in
terms of the reproducing kernel of $\mathcal{A}\left( \Omega _{1,R}\right) $%
. The latter was obtained in \cite{Pee-Zha} (see section 7 there) and may be
written in a more compact form using Euler's reflection formula for the
Gamma function \cite{AAR} as: 
\begin{equation}
K_{0}^{(R,B)}\left( z,w\right) :=\frac{(2\pi)^{2B-3}}{\Gamma \left(
2B-1\right) R^{B}(\log R)^{2B-1}}\sum\limits_{j\in \mathbb{Z}}\left\vert \Gamma
\left( B+i\frac{\log R}{\pi }\left( j+B\right) \right) \right\vert
^{2}\left( \frac{z\overline{w}}{R}\right) ^{j}.  \label{RKA}
\end{equation}%
The subspace $\mathcal{A}\left( \Omega _{1,R}\right) $ fits also the null space 
\begin{equation}
\mathcal{E}_{0}\left( \Omega _{1,R}\right) =\left\{ \phi \in \mathfrak{H}%
_{B}\left( \Omega _{1,R}\right) ,\text{ }\Delta _{B}\phi =0\text{ }\right\} 
\label{null}
\end{equation}%
of the following so-called invariant Laplacian operator with weight $B$ t(\cite{Pee}): 
\begin{equation}
\Delta _{B}:=-\left( \omega _{R}\left( z\right) \right) ^{2}\partial _{z%
\overline{z}}-2B\omega _{R}\left( z\right) \left( \partial _{z}\omega
_{R}\left( z\right) \right) \partial _{\overline{z}}.  \label{WL}
\end{equation}%
This is a densely-defined operator on $\mathfrak{H}_{B}\left( \Omega
_{1,R}\right) $ whose discrete spectrum consists of the following finite set
of eigenvalues: 
\begin{equation*}
\lambda _{B,m}:=-m\left( 2B-m-1\right) ,\text{ }m=0,1,...,\left\lfloor
B-\left( 1/2\right) \right\rfloor ,
\end{equation*}%
each being of infinite multiplicity ($\left\lfloor x\right\rfloor $ stands
for the greatest integer not exceeding $x$).

In this paper, we are concerned with higher Landau-Levels eigenspaces: 
\begin{equation}  \label{HLE}
\mathcal{E}_{m}\left( \Omega _{1,R}\right) =\left\{ \phi \in \mathfrak{H}%
_{B}\left( \Omega _{1,R}\right) ,\Delta _{B}\phi =\lambda _{B,m}\phi
\right\}, \quad m=0,1,...,\left\lfloor B-\left( 1/2\right) \right\rfloor.
\end{equation}%
For a fixed Landau level $\lambda _{B,m}$, the corresponding eigenspace
turns out to be the $m$th true polyanalytic space $\Omega_{1,R}$ with
respect to the invariant Cauchy-Riemann operator $\left( \omega _{R}\left(
z\right)\right) ^{2}\overline{\partial }$. Moreover, we shall extend
Peetre's formula \eqref{RKA} by establishing an explicit formula for the
reproducing kernel $K_{m}^{(R,B)}\left( z,w\right)$ of the Hilbert space $%
\mathcal{E}_{m}\left( \Omega _{1,R}\right)$. More precisely, the
non-orthonormal basis elements were expressed in \cite{Pee} through
Routh-Romanovski polynomials \cite{Rom,Rou} and we shall compute below their 
$L^2$-norm. These eigenfunctions were determined in \cite{Pee} after
carrying the eigenvalue problem into the one of the Schr\"{o}dinger operator
associated with the hyperbolic Scarf potential \cite{AK}. In this respect,
it is worth noting that the boundedness of this potential implies the
existence of a continuous spectrum for $\Delta _{B}$ corresponding to
scattering states.

Back to bound states, it is known that Routh-Romanovski polynomials may be
represented through Jacobi polynomials with imaginary arguments and
parameters. Making use of this relation, we shall prove that under the
quantization condition $B\in \mathbb{Z}_{+}$, the reproducing kernel $%
K_{m}^{(R,B)}\left( z,w\right) $ may be expressed through higher derivatives
of the fourth Jacobi's theta function $\theta _{4} $. In
particular, one retrieves the known fact that $K_{0}^{(R,B)}\left(
z,w\right) $ is closely connected to Weierstrass elliptic function $\wp $
associated with the rectangular lattice \cite{Ber}. Under the same
condition, we shall also derive the transformation rule of $%
K_{m}^{(R,B)}\left( z,w\right) $ under the action of the inversion with
respect to the circle centered at the origin and of radius $R$. Since this
kernel is readily seen to be invariant under rotations, this transformation
rule exhausts its quasi-invariance under the automorphism group of the
annulus. Of course, one can not expect a
strong analogy with the Poincar\'{e} disc case since the automorphism group
of the latter model is much more larger than the one of the annulus.

The paper is organized as follows. In section 2, we recall some geometrical
facts about the invariant Laplacian $\Delta _{B}$ as well as some of its
needed spectral properties. In section 3, we write down the orthonormal
basis of the eigenspaces \eqref{HLE} associated with the discrete spectrum
of $\Delta _{B}$, discuss their poly-analyticity property and derive
explicit expressions for the corresponding reproducing kernels. Section 4
is devoted to the relation of these kernels to the fourth Jacobi Theta
functions and to their invariance properties under the quantization
condition. Section 5 contains concluding remarks with a particular emphasis
on probabilistic aspects of reproducing kernels of poly-analytic Hilbert
spaces. Proofs of our results are detailed in three appendices.

\section{$L^{2}$ spectral theory of $\Delta _{B}$}

The Poincar\'{e} metric of the annulus $\Omega_{1,R}$ is written in local coordinates as: 
\begin{equation*}
ds=\frac{|dz|}{\omega _{R}\left( z\right) },
\end{equation*}%
where $\omega _{R}\left( z\right) $ is defined in $\left( 1.2\right) $. It allows to define the $\left( 1,0\right) -$ connection \textit{\ \ 
}%
\begin{equation*}
\varpi =\partial _{z}+B\partial _{z}\left( \log \omega _{R}\left( z\right)
\right) 
\end{equation*}%
to which is associated the Bochner Laplacian: 
\begin{equation*}
H_{B}=-(\omega _{R})^{2}(z)\left( \partial _{z}+B\partial _{z}\log (\omega
_{R}\left( z\right) )\right) \left( \partial _{\overline{z}}-B\partial _{%
\overline{z}}\log (\omega _{R}\left( z\right) )\right) 
\end{equation*}%
\begin{equation*}
=-(\omega _{R}\left( z\right) )^{2}\left( \partial _{\overline{z}}-B\partial
_{\overline{z}}\log (\omega _{R})\right) \left( \partial _{z}+B\partial
_{z}\log (\omega _{R}\left( z\right) )\right) +2B(\omega _{R}\left( z\right)
)^{2}\partial _{z\overline{z}}\log (w_{R}\left( z\right) ).
\end{equation*}%
This is a densely-defined operator on the weighted $L^{2}$-space 
\begin{equation*}
\mathfrak{H}_{0}\left( \Omega _{1,R}\right) =L^{2}\left( \Omega
_{1,R},\left( \omega _{R}\left( z\right) \right) ^{-2}d\mu \left( z\right)
\right) .
\end{equation*}%
Moreover, by analogy with (\cite{Shige}, p.124), the \textit{ground state
transformation} 
\begin{equation*}
Q_{B}:\mathfrak{H}_{0}\left( \Omega _{1,R}\right) \rightarrow \mathfrak{H}%
_{B}\left( \Omega _{1,R}\right) 
\end{equation*}%
defined by: 
\begin{equation*}
Q_{B}\left[ f\right] \left( z\right) =\left( \omega _{R}\left( z\right)
\right) ^{-B}f\left( z\right) ,\text{ \ \ \ \ }z\in \Omega _{1,R},
\end{equation*}%
is a unitary map and intertwines the operators $H_{B}$ and $\Delta _{B}$: 
\begin{equation*}
\left( Q_{B}\right) ^{-1}\circ H_{B}\circ Q_{B}=\Delta _{B}.
\end{equation*}%
On the other hand, since $\omega _{R}\left( z\right) $ is radial then $%
\Delta _{B}$ may be mapped to a Sturm-Liouville operator by letting it act
on functions of the form: 
\begin{equation*}
z\mapsto z^{j}f\left( \cot \zeta \right) ,\text{ }\zeta :=\frac{\pi }{\log R}%
\log \left\vert z\right\vert \in \left( 0,\pi \right) ,\text{ }j\in \mathbb{Z%
}\text{.}
\end{equation*}%
Doing so leads to the second-order differential operator in the variable $%
\xi =\cot \zeta ,$ given by (\cite{Pee}): 
\begin{equation*}
\mathcal{L}_{B}:=\left( 1+\xi ^{2}\right) \frac{d^{2}}{d\xi ^{2}}+2\left[
\left( 1-B\right) \xi -\left( j+B\right) \frac{\log R}{\pi }\right] \frac{d}{%
d\xi }
\end{equation*}%
whose eigenfunctions are given by Routh-Romanovski polynomials \cite{RWAK}: 
\begin{equation*}
\mathcal{R}_{m}^{\left(2\left( j+B\right) (\log R)/\pi,1-B\right)}\left( \xi \right) ,\text{ \ }m=1,2,...,\left\lfloor B-(1/2)\right\rfloor.
\end{equation*}%
These polynomials may be represented through Jacobi polynomials with
complex-conjugate imaginary parameters (\cite{Martinez}, p.2) : 
\begin{equation}
\mathcal{R}_{k}^{\left( a,b\right) }\left( x\right) =\left( -2i\right)
^{k}k!P_{k}^{\left( b-1+\frac{1}{2}ia,b-1-\frac{1}{2}ia\right) }\left(
ix\right) ,\text{ \ }k=0,1,...,\text{ \ }i^{2}=-1.  \label{Routh}
\end{equation}%
Set%
\begin{equation}\label{Alpha}
\alpha \left( j,B\right) := \frac{2}{\pi}\left( j+B\right) \log R 
\end{equation}%
then the Routh-Romanovski polynomials above are finitely orthogonal with respect to the Student-type weight: 
\begin{equation}
\varrho _{j}^{R,B}\left( \xi \right) :=\left( 1+\xi ^{2}\right) ^{B}\exp
\left( \alpha \left( j,B\right) \cot ^{-1}\left( \xi \right) \right) ,\text{%
{}}\xi \in \mathbb{R},  \label{Student}
\end{equation}%
whose moments exists up to the order $\left\lfloor 2B-1\right\rfloor $. In addition, $\mathcal{L}_{B}$ may be transformed to the Schr\"{o}dinger operator with the hyperbolic Scarf potential\footnote{The details of this transformation are written in \cite{Pee} but there are some misprints and the author missed a term proportional to $1/\cosh^2(\theta), \xi = \sinh(\theta)$.} , also referred to as Scarf II (\cite{AK}). 
Since this potential is bounded, then $\mathcal{L}_{B}$ admits a
continuous spectrum whose eigenfunctions (scattering states) may be found in \cite{KS}. 

\begin{remark}
The Routh-Romanovski polynomials were discovered by J. Routh \cite{Rou} and
rediscovered by V. I. Romanovski \cite{Rom} within the context of
probability distributions. They are also named Romanovski of type IV or
\textquotedblleft finite Romanovski\textquotedblright due to their finite-orthogonality.
\end{remark}

\begin{remark}
The spectral theory of $\Delta_B$ bears some similarities with that of the
Schr\"{o}dinger operator with uniform magnetic filed in the Poincar\'{e}
upper half-plane $\mathbb{H}^{2}=\left\{ \left( x,y\right) ,x\in \mathbb{R}%
,y>0\right\}$. The latter is actually given in suitable units by $%
(-M_{B}+B^{2})/2,$ where 
\begin{equation*}
M_{B} :=y^{2}\left( \partial _{x}^{2}+\partial_{y}^{2}\right) -2iBy\partial
_{x}
\end{equation*}
is the $B$-weight Maass Laplacian \cite{Mou1, Mou2}. Moreover, up to the
variable change $s=-\log y,$ and when acting on functions of the form $%
g\left( s,x\right) :=\exp \left(-i\gamma x-\frac{1}{2}s\right) \psi \left(
s\right) ,\gamma \in \mathbb{R}$, this operator is mapped to the
one-dimensional Schr\"{o}dinger operator with Morse potential (\cite{Morse}%
): 
\begin{equation*}
-\frac{1}{2}\partial _{s}^{2}g+\left( \frac{1}{2}\gamma ^{2}\exp \left(
-2s\right) +\gamma B\exp \left( -s\right) \right)g+\left( \frac{1}{2}B^{2}+%
\frac{1}{8}\right) g.
\end{equation*}%
For more details on this connection, we refer the reader to \cite{Linet,
Mou3}. The structure of the spectrum is also similar: provided that $B>1/2$,
a finite number of eigenvalues usually known as hyperbolic Landau levels
arises. Physically, this phenomenon means that the magnetic field has to be
strong enough to capture the particle in a closed orbit, giving rise to 
\textit{bound states} in which the particle cannot leave the system without
additional energy (\cite{Com}).
\end{remark}

\section{Reproducing kernel of $\mathcal{E}_{m}(\Omega _{1,R})$}

Let $B>1/2$ and fix $m=0,1,...,\left\lfloor B-1/2\right\rfloor $. Then an
orthogonal basis of the Hilbert space $\mathcal{E}_{m}(\Omega _{1,R})$ is
given by (\cite{Pee}): 
\begin{equation}\label{Basis}
\phi _{j}\left( z\right) =z^{j}\mathcal{R}_{m}^{\left( \alpha \left(
j,B\right) ,1-B\right) }\left( \cot \left( \frac{\pi }{\log R}\log
\left\vert z\right\vert \right) \right) ,\,z\in \Omega _{1,R},\text{ }j\in 
\mathbb{Z}.
\end{equation}%
The squared norm of $\phi _{j}$\ in $\mathfrak{H}_{B}\left( \Omega
_{1,R}\right) $ admits the following expression (see Appendix A for the
proof): 
\begin{equation}
\left\Vert \phi _{j}\right\Vert ^{2}=2^{3-2(B-m)}R^{B}\frac{(\log R)^{2B-1}}{\pi^{2B-3}}\frac{%
m!\Gamma (2B-m)}{(2(B-m)-1)}\frac{R^{j}}{\left\vert \Gamma \left( B-m+\frac{1%
}{2}i\alpha \left( j,B\right) \right) \right\vert ^{2}}.  \label{SQN}
\end{equation}%
Furthermore, these basis elements satisfy the poly-analyticity property with
respect to the invariant Cauchy-Riemann operator (see \cite{Pee-Zha1},
pp.241-243): 
\begin{equation*}
D_{\overline{z}}^{\omega }:=\left( \omega \left( z\right) \right)
^{2}\partial _{\overline{z}}.
\end{equation*}%
In particular, for $m=0,$ the null space $\mathcal{E}_{0}(\Omega _{1,R})$
coincides with the space $\mathcal{A}(\Omega _{1,R})$ of analytic
functions in $\Omega _{1,R}$ belonging to $\mathfrak{H}_{B}\left( \Omega
_{1,R}\right) $. However, one readily checks by direct computations that
unless $m=0$, we have 
\begin{equation*}
\left( \partial _{\overline{z}}\right) ^{m+1}\mathcal{R}_{m}^{(\alpha
,1-B)}\left( \cot \left( \frac{\pi \ln (|z|)}{\ln (R)}\right) \right) \neq 0.
\end{equation*}%
More generally, denote 
\begin{equation*}
\mathcal{F}^{\left( m\right) }(\Omega _{1,R}):=\left\{ \phi \in \mathfrak{H}%
_{B}\left( \Omega _{1,R}\right) ,\left( D_{\overline{z}}^{\omega }\right)
^{m+1}\left[ \phi \right] =0\right\} 
\end{equation*}%
the polyanalytic space of order $m$. Then, the eigenspace in \eqref{HLE} may be decomposed as: 
\begin{equation*}
\mathcal{E}_{m}(\Omega _{1,R})=\mathcal{F}^{\left( m+1\right) }(\Omega
_{1,R})\ominus \mathcal{F}^{\left( m\right) }(\Omega _{1,R}),
\end{equation*}%
where $\ominus $ stands for the orthogonal difference of two sets. This fact
is a direct consequence of the factorization property proved in \cite%
{Pee-Zha1} and valid for arbitrary hyperbolic Riemanns surfaces. Accordingly, $%
\mathcal{E}_{m}(\Omega _{1,R})$ will be referred to as the $m$-th
true-polyanalytic space on the annulus $\Omega _{1,R}$ with respect to $D_{%
\overline{z}}^{\omega }.$ In appendix B, we shall prove the following
formula for the reproducing kernel of this eigenspace.

\begin{Thm}
\label{RepKer} Let $B>\frac{1}{2}$, $m=0,1,...,\left\lfloor B-\frac{1}{2}\right\rfloor$. Then, the reproducing kernel of the $m$-th eigenspace $\mathcal{%
E}_{m}(\Omega _{1,R})$ reads 
\begin{equation*}
K_{m}^{R,B}\left( z,w\right) =\frac{(2\pi)^{2B-3}(2B-2m-1)}{R^{B}(\log R)^{2B-1}\Gamma (2B-m)}\sum\limits_{l=0}^{m}\sum\limits_{k=0}^{m-l}\frac{\left(
1-2B+m\right) _{k+l}}{(m-k-l)!}\frac{V^{k}\overline{V}^{l}}{k!l!}\sigma
_{k,l}^{R,B}(z,w)
\end{equation*}%
where 
\begin{equation}
\sigma _{k,l}^{R,B}(z,w)=\sum_{j\in \mathbb{Z}}\Gamma \left( B-k+i\alpha(j,B)/2\right) \Gamma \left(B-l-i\alpha(j,B)/2\right) \left( \frac{z\overline{w}}{R}\right) ^{j}  \label{Not1}
\end{equation}%
and 
\begin{equation}\label{V}
V=\frac{1}{4}\left( 1+i\cot \left( \frac{\pi \log (|z|)}{\log (R)}\right)
\right) \left( 1+i\cot \left( \frac{\pi \log (|w|)}{\log (R)}\right) \right) 
\end{equation}%
for every $z,w\in \Omega _{1,R}.$
\end{Thm}

When $m=0$, we recover the (analytic) reproducing kernel $%
K_{0}^{(R,B)}\left( z,w\right) $ in \eqref{RKA}. This is in
agreement with the computations done in (\cite{Pee-Zha}, p.263), if we
identify the parameter $\alpha $ there with $2B-2$ here and if we use Euler
reflection's formula: 
\begin{equation*}
\Gamma (z)\Gamma (1-z)=\frac{\pi }{\sin (\pi z)}.
\end{equation*}%
In particular, if $B \geq 1$ is an integer then the index change $j+B \rightarrow j$ together with formula 8 from \cite{Erd}, p.4: 
\begin{equation*}
\Gamma \left(B+ij\frac{\log R}{\pi }\right) \Gamma \left(B-ij\frac{\log R}{\pi }\right) =\frac{2\log (R)[\Gamma (B)]^{2}jR^{j}}{R^{2j}-1}\prod_{q=1}^{B-1}\left( 1+\frac{(j\log R)^{2}}{\pi ^{2}q^{2}}\right)
\end{equation*}%
yield 
\begin{equation*}
K_{0}^{(R,B)}\left( z,w\right) =\frac{(2\pi)^{2B-2}[\Gamma (B)]^{2}}{\pi \Gamma \left(2B-1\right)(z\overline{w})^B(\log R)^{2B-2}}\sum\limits_{j\in \mathbb{Z}} \frac{j}{R^{2j}-1}\prod_{q=1}^{B-1}\left( 1+\frac{(j\log R)^{2}}{\pi ^{2}q^{2}}\right)
\left(z\overline{w}\right) ^{j}. 
\end{equation*}%
As a result, the following limit holds:
\begin{equation*}
\lim_{R \rightarrow \infty} K_{0}^{(R,B)}\left( z,w\right) =\frac{2^{2B-2}}{\pi\Gamma \left(2B-1\right)}\sum\limits_{j\geq 1} \frac{j^{2B-1}}{\left(z\overline{w}\right) ^{j+B}},
\end{equation*}
which is the polyanalytic reproducing kernel of the complementary of the closed unit disc (or after the variable change $z\overline{w} \rightarrow 1/(z\overline{w})$ the one of the punctured open unit disc). 

On the other hand, if $B=1$ then:  
\begin{equation}
K_{0}^{(R,1)}(z,w)=\frac{1}{\pi z\overline{w}}\sum_{j\in \mathbb{Z}}\frac{j}{%
1-1/R^{2j}}\left( \frac{z\overline{w}}{R^{2}}\right) ^{j}.  \tag{3.4}
\end{equation}%
Since $z/R,y/R$ belong to the annulus ${\Omega }_{1/R,1}$, then we retrieve
the known formula of the analytic reproducing kernel in ${\Omega }_{1/R,1}$
endowed with its Lebesgue measure which may be expressed through the
Weierstrass elliptic function $\wp $ \cite{Ber}. Below, we shall
prove  that if $B\geq 1$ is an integer, then $%
K_{m}^{(R,B)}(z,y)$ is a linear combination of higher logarithmic
derivatives of Jacobi's fourth Theta function.

\section{Some properties of $K_{m}^{R,B}\left( z,w\right) $}

In this section, we shall assume that $B \geq 1$ is integer. Thinking of $B$ as the flux of the (closed two-form) magnetic field, this integrality assumption is indeed a quantization condition under which the poly-analytic Bergman kernel may be related to the fourth Jacobi's Theta function and satisfies a transformation rule under the inversion $\mathcal{I}_{R}: z \mapsto R/z$. 

\subsection{Poly-analytic Bergman kernel and fourth Jacobi Theta function}
Let 
\begin{equation*}
\theta _{4}\left( z,\tau \right) =1+2\sum\limits_{k=0}^{\infty }\left(
-1\right) ^{k}\tau ^{-k^{2}}\cos (2kz).
\end{equation*}
be the fourth theta function de Jacobi (\cite{Erd1}, p.355). Then, we shall prove in Appendix C the following result:
\begin{propo}
\label{Theta} For any $m=0,1,\cdots ,\left\lfloor B-1/2\right\rfloor$ the
poly-analytic reproducing kernel $K_{m}^{R,B}(z,w)$ may be written as a finite sum
of higher logarithmic derivatives of the fourth Jacobi's Theta function $%
\theta _{4}(z,R)$ associated with the rectangular lattice parameter $\tau
=i\log (R)/\pi$ and evaluated at the point $(i/2)\log (z\overline{w}/R)$.
\end{propo}

\subsection{Invariance under the automorphism group}
The automorphism group of the annulus ${\Omega }_{1,R}$ the direct product
of the rotation group and of the cyclic group generated by the inversion $\mathcal{I}_{R}$. It is much
more smaller than the Mobius group and does not act transitively on ${\Omega 
}_{1,R}$. While the kernel $K_{m}^{(R,B)}\left( z,w\right) $ is obviously
seen to be rotation-invariant (see \eqref{B3} in Appendix B below) , its transformation rule under inversion is not clear
unless $B\geq 1$ is an integer. Indeed, the index change $j\mapsto j-B$ in \eqref{B3}
shows in this case that 
\begin{equation*}
K_{m}^{(R,B)}\left( \frac{R}{z},\frac{R}{w}\right) =\left( \frac{z\overline{w%
}}{R}\right) ^{2B}K_{m}^{(R,B)}(z,w),
\end{equation*}%
or equivalently 
\begin{equation*}
K_{m}^{(R,B)}(z,w)=\left[ (\mathcal{I}_{R})^{\prime }(z)(\mathcal{I}%
_{R})^{\prime }\left( w\right) \right] ^{B}K_{m}^{(R,B)}\left( \mathcal{I}%
_{R}\left( z\right) ,\mathcal{I}_{R}\left( w\right) \right) .
\end{equation*}%
Written in this form, this transformation rule reminds the one satisfied by
the reproducing kernel of the hyperbolic disc under the action of the M\"{o}%
bius group (see e.g. \cite{hayouss}).\medskip

For general values of $B>1/2$, it is readily seen from the expression \eqref{B5} below that the inversion $\mathcal{I}_{R}$ has the effect to transform $\alpha (j,B)$ to $\alpha (j,-B)$ which is meaningless since $B$ is positive. 

\section{Concluding remarks}

In this paper, we derived the reproducing kernels corresponding to
eigenspaces of the magnetic Laplacian on the annulus. We
also expressed them by means of the fourth Jacobi theta function when the
magnetic flux is an integer and investigated their transformation
rule under the automorphism group of the annulus. At the probabilistic side,
reproducing kernels of Hilbert spaces provide very interesting examples of
determinantal point processes (DPP). For instance, for $l=1,2,...,$ the
kernel $(1-\zeta \overline{\zeta ^{\prime }})^{l+1},\zeta ,\zeta ^{\prime }\in 
\mathbb{D},$ governs the determinantal correlation functions of the zeros of entire
series whose coefficients are independent $l\times l$ Ginibre matrices \cite%
{Per-Vir,Kri}. This example was generalized in \cite{Dem-Laz} using the
magnetic Laplacian in the hyperbolic disc, however no connection to random
matrices was found yet in general. The flat counterpart of this DPP is
related to the Fock space and was introduced and studied in \cite{Shi},
extending the celebrated Ginibre point process. Remarkably, the reproducing kernel
of the hyperbolic DPP studied in \cite{Dem-Laz} is polyanalytic with respect
to both the Wirtinger operator $\partial _{\overline{z}}$ (\cite{hayouss})
and its weighted counterpart (that is the invariant Cauchy-Riemann operator, 
\cite{Pee-Zha1}). Though this property is obvious in the flat setting since
both operators coincide, it remains intriguing for the hyperbolic Disc
geometry since powers of both operators are clearly different (see \cite%
{Eng-Zha} for a comparison of these operators). By direct calculations, one checks that the poly-analyticity property with respect
to $\partial _{\overline{z}}$ fails for the basis elements of $\mathcal{E}_{m}\left( \Omega _{1,R}\right) ,m\neq 0$ while it still holds true for the
invariant Cauchy-Riemann operator. Finally, we would like to point out the
recent preprint \cite{Kat-Shi} where the authors relate zeroes of Laurent
series with Gaussian coefficients to a hyper-determinantal point process
governed by the Szeg\"o kernel of the annulus (see also \cite{Kat} for other
connections of DPP to elliptic functions).

\begin{center}
\textbf{Appendix A}
\end{center}

\begin{proof}[Proof of \eqref{SQN}]
Recall the squared $L^{2}$-norm with respect to the weight $\left(\omega _{R}\left( z\right) \right) ^{2B-2}$: 
\begin{equation*}
\left\Vert \phi _{j}\right\Vert ^{2}=\int\limits_{\Omega _{1,R}}\phi
_{j}\left( z\right) \overline{\phi _{j}\left( z\right) }\left( \omega
_{R}\left( z\right) \right) ^{2B-2}d\mu \left( z\right)
\end{equation*}%
where $\phi_j$ is given by \eqref{Basis}. Using elementary variable changes, this integral is expressed as the $L^{2}$-norm of the $m$-th Routh-Romanovski polynomial with respect to the
orthogonality weight \eqref{Student} as%
\begin{equation}  \label{A2}
\left\Vert \phi _{j}\right\Vert ^{2}=2\frac{(\log R)^{2B-1}}{\pi^{2B-2}}\int\limits_{\mathbb{-\infty }%
}^{+\infty }\mathcal{R}_{m}^{\left( \alpha \left( j,B\right) ,1-B\right)
}\left( \xi \right) \varrho _{j}^{R,B}\left( \xi \right) d\xi.
\end{equation}%
Now, the Routh-Romanovski polynomials admits the Rodrigues representation (\cite{Pe}, p.259): 
\begin{equation*}
\mathcal{R}_{m}^{\left( \alpha ,\beta \right) }\left(\xi\right) =\frac{1}{\omega ^{\left( \alpha ,\beta \right) }\left( \xi\right) }\frac{d^{m}}{d\xi^{m}}%
\left( \omega ^{\left( \alpha ,\beta \right) }\left(\xi\right) \left(1+\xi^{2}\right) ^{k}\right) ,\text{ \ }
\end{equation*}%
with%
\begin{equation*}
\omega ^{\left( \alpha ,\beta \right) }\left( \xi\right) := e^{-\alpha\cot^{-1}\xi}\left( 1+\xi^{2}\right) ^{\beta -1}.
\end{equation*}%
Substitute the polynomial $\mathcal{R}_{m}^{\left(\alpha \left(j,B\right) ,1-B\right) }\left( \xi \right)$ in \eqref{A2} by its Rodrigues representation and performing $m$ integration by parts, one gets: 
\begin{equation}  \label{A5}
\left\Vert \phi _{j}\right\Vert ^{2}=\left( -1\right) ^{m}m!a_{m}^{\left(
\alpha \left( j,B\right) ,1-B\right) }\frac{(\log R)^{2B-1}}{\pi^{2B-2}} \int\limits_{\mathbb{-\infty }%
}^{+\infty }\varrho _j^{R,B}\left( \xi \right) \left( 1+\xi ^{2}\right)
^{m}d\xi
\end{equation}%
where $a_{m}^{\left( \alpha \left( j,B\right) ,1-B\right) }$ is the leading
coefficient of $\mathcal{R}_{m}^{\left( \alpha \left( j,B\right) ,1-B\right)
}\left( \xi \right) ,$ which depends only on $B$ and $m.$ Indeed, the Routh-Romanovski and the Jacobi polynomials with complex-conjugate imaginary parameters are
interrelated via (\cite{Martinez}, p.2): 
\begin{align*}
\mathcal{R}_{m}^{(\alpha(j,B) ,1-B)}(x) &=(-2i)^{m}m!P_{m}^{(-B+i\alpha(j,B)/2,-B-i\alpha(j,B)/2)}(ix)
\\& =(-2i)^{m}m!\sum_{l=0}^{m}\frac{(m-2B+1)_{l}(-B+i\alpha(j,B)/2+l+1)_{m-l}}{l!(m-l)!}\frac{(ix-1)^{l}}{2^{l}},
\end{align*}%
whence we readily get the following expression for the leading coefficient: 
\begin{equation}\label{LC}
a_{m}^{\left( \alpha(j,B),1-B\right) }=(-1)^{m}\frac{\Gamma (2B-m)}{%
\Gamma(2B-2m)}.
\end{equation}%
Now, the variable change $\xi =\cot \theta $ yields:
\begin{align}\label{A6}
\int\limits_{\mathbb{-\infty }}^{+\infty }\left( 1+\xi ^{2}\right)^{m}\varrho _{j}^{R,B}\left( \xi \right) d\xi & =\int\limits_{0}^{\pi}e^{\alpha(j,B)\theta }\left( \sin \theta \right) ^{^{2\left( B-m\right)-2}}d\theta \nonumber
\\& =\frac{\pi e^{\pi\alpha(j,B)/2}\Gamma \left( 2\left( B-m\right) -1\right) }{2^{2\left( B-m\right) -2}\Gamma \left( B-m+i\alpha _{j}/2\right) \Gamma\left( B-m-i\alpha _{j}/2\right) }
\end{align}%
where we applied the Cauchy Beta integral (\cite{PBM1} ,\text{ p.445}, \cite%
{MMH}): 
\begin{equation*}
\int\limits_{0}^{\pi }e^{-p\,x}(\sin x)^{\nu }\,dx=\frac{2^{-\nu }\pi
e^{-\pi p/2}}{(\nu +1)\mathbf{B}\left( \frac{1}{2}\left( \nu +ip\right) +1,\,%
\frac{1}{2}\left( \nu -ip\right) +1\right) },\qquad \nu >-1,
\end{equation*}%
with $p=-\alpha(j,B)$ and $\nu =2\left(B-m\right) -2.$. Combining \eqref{A2}, \eqref{A5} and \eqref{A6}, we end up with: 
\begin{equation*}
\left\Vert \phi _{j}\right\Vert ^{2}= \frac{(\log R)^{2B-1}}{\pi^{2B-3}}2^{3-2(B-m)}\left(
-1\right) ^{m}m!a_{m}^{\left( \alpha \left( j,B\right) ,1-B\right) }\frac{%
e^{\pi\alpha \left( j,B\right)/2 }\Gamma \left( 2\left( B-m\right) -1\right) }{%
\left\vert \Gamma \left( B-m+i\alpha \left( j,B\right) /2\right) \right\vert^{2}},
\end{equation*}
Keeping in mind \eqref{Alpha} and \eqref{LC}, the expression \eqref{SQN} of the squared norm follows. 
\end{proof}

\begin{center}
\textbf{Appendix B}
\end{center}

\begin{proof}[Proof of Theorem \protect\ref{RepKer}]
According to \eqref{SQN}, an orthonormal basis for $\mathcal{E}_{m}(\Omega_{1,R})$ given by: 
\begin{multline}
\Phi _{j}\left( z\right) :=\left( \frac{(\log R)^{2B-1}}{\pi^{2B-3}}\frac{m!\Gamma (2B-m)}{(2(B-m)-1)}\frac{2^{3-2(B-m)}R^{j+B}}{\left\vert \Gamma \left( B-m+i\left( \frac{\log R}{\pi }%
\right) \left( j+B\right) \right) \right\vert ^{2}}\right) ^{-1/2}
\label{B1} \\
z^{j}\mathcal{R}_{m}^{\left( \alpha \left( j,B\right) ,1-B\right) }\left(
\cot \left( \frac{\pi \log \left\vert z\right\vert }{\log R}\right) \right) ,%
\text{ }z\in \Omega _{1,R},\text{ }j\in \mathbb{Z}.
\end{multline}%
\ Then, by the general theory of reproducing kernels (\cite{Saitoh}, p.119): 
\begin{equation}
K_{m}^{R,B}\left( z,w\right) =\sum\limits_{j\in \mathbb{Z}}\Phi _{j}\left(
z\right) \overline{\Phi _{j}\left( w\right) }.  \label{B2}
\end{equation}%
Inserting \eqref{B1} into \eqref{B2}, we get: 
\begin{multline}
K_{m}^{R,B}\left( z,w\right) =\frac{\pi^{2B-3}2^{2(B-m)-3}\left( 2\left( B-m\right)
-1\right) }{R^{B}m!(\log R)^{2B-1}\Gamma \left( 2B-m\right) }\sum_{j\in \mathbb{%
Z}}\left\vert \Gamma \left( B-m+i\alpha \left( j,B\right) /2\right)
\right\vert ^{2}\left( \frac{z\overline{w}}{R}\right) ^{j}  \label{B3} \\
\mathcal{R}_{m}^{\left( \alpha \left( j,B\right) ,1-B\right) }\left( \cot
\left( \frac{\pi \log \left\vert z\right\vert }{\log R}\right) \right) 
\overline{\mathcal{R}_{m}^{\left( \alpha \left( j,B\right) ,1-B\right)
}\left( \cot \left( \frac{\pi \log \left\vert w\right\vert }{\log R}\right)
\right) }. 
\end{multline}%
Using again the relation (\cite{Martinez}, p.2): 
\begin{equation*}
\mathcal{R}_{m}^{\left( a,b\right) }\left( x\right) =\left( -2i\right)
^{m}m!P_{m}^{\left( b-1+ia/2,b-1-ia/2\right) }\left( ix\right) ,\quad m\geq
0,
\end{equation*}%
then \eqref{B3} takes the form 
\begin{multline}
K_{m}^{(R,B)}(z,w)=\frac{(2\pi)^{2B-3}m!(2B-2m-1)}{R^{B}(\log R)^{2B-1}\Gamma (2B-m)}%
\sum_{j\in \mathbb{Z}}\left( \frac{z\overline{w}}{R}\right) ^{j}|\Gamma (B-m+%
\frac{i}{2}\alpha \left( j,B\right) )|^{2}  \label{B5} \\
P_{m}^{(-B+\frac{i}{2}\alpha \left( j,B\right) ,-B-\frac{i}{2}i\alpha \left(
j,B\right) )}\left( i\cot \left( \frac{\pi \log |z|}{\log R}\right) \right)
P_{m}^{(-B-\frac{i}{2}\alpha \left( j,B\right) ,-B+\frac{i}{2}\alpha \left(
j,B\right) )}\left( -i\cot \left( \frac{\pi \log |w|}{\ln R}\right) \right) .
\end{multline}%
For sake of simplicity, we introduce the following notations: 
\begin{equation*}
\mu _{j}:=\frac{i}{2}\alpha \left(j,B\right) = \frac{i}{\pi}(j+B)\log(R)   ,\quad t=\frac{z\overline{w}%
}{R},\text{ }
\end{equation*}%
and%
\begin{equation*}
X=\cot \left( \frac{\pi \log |z|}{\log R}\right) ,\text{ }Y=\cot \left( 
\frac{\pi \log |w|}{\log R}\right) .
\end{equation*}%
Consequently, \eqref{B5} reads: 
\begin{equation}\label{B6}
K_{m}^{(R,B)}(z,w)=\gamma _{m}^{R,B}\sum_{j\in \mathbb{Z}}t^{j}|\Gamma
(B-m+\mu _{j})|^{2}P_{m}^{(-B+\mu _{j},-B-\mu _{j})}\left( iX\right)
P_{m}^{(-B-\mu _{j},-B+\mu _{j})}\left( -iY\right) ,
\end{equation}%
where 
\begin{equation}
\gamma _{m}^{R,B}=\frac{(2\pi)^{2B-3}m!\left( 2B-2m-1\right) }{R^{B}(\log R)^{2B-1}\Gamma (2B-m)}.
\end{equation}%
Equivalently, the symmetry relation (\cite{AAR}, p.305) 
\begin{equation*}
P_{m}^{(-B-\mu _{j},-B+\mu _{j})}\left( -iY\right) =\left( -1\right)
^{m}P_{m}^{(-B+\mu _{j},-B-\mu _{j})}\left( iY\right) 
\end{equation*}%
entails 
\begin{equation*}
K_{m}^{(R,B)}(z,w)=\left( -1\right) ^{m}\gamma _{m}^{R,B}\sum_{j\in \mathbb{Z%
}}t^{j}|\Gamma (B-m+\mu _{j})|^{2}P_{m}^{(-B+\mu _{j},-B-\mu _{j})}\left(
iX\right) P_{m}^{(-B+\mu _{j},-B-\mu _{j})}\left( iY\right) .
\end{equation*}%
Now, recall Bateman's formula (\cite{Bat}\text{, p.392}): 
\begin{multline*}
P_{m}^{(\alpha ,\beta )}\left( x\right) P_{m}^{(\alpha ,\beta )}\left(
y\right) =\sum\limits_{k=0}^{m}\left( -1\right) ^{m+k}\frac{\left( \alpha
+\beta +m+1\right) _{k}}{m!\left(m-k\right) !}\left( \frac{x+y}{2}\right)
^{k} \\
\times \frac{\Gamma \left( \alpha +m+1\right) \Gamma \left( \beta
+m+1\right) }{\Gamma \left( \alpha +k+1\right) \Gamma \left( \beta
+k+1\right) }P_{k}^{\left( \alpha ,\beta \right) }\left( \frac{1+xy}{x+y}%
\right) 
\end{multline*}%
as well as the expression of the Jacobi polynomial (\cite{Szego1939}, p.67) 
\begin{equation*}
P_{k}^{(\alpha ,\beta )}\left( x\right) = \sum\limits_{l=0}^{k}\binom{k+\alpha}{k-l}\binom{k+\beta}{l} \left( \frac{x-1}{2}\right) ^{l}\left( \frac{x+1}{2}\right) ^{k-l} 
\end{equation*}%
to write 
\begin{eqnarray}
P_{m}^{(\alpha ,\beta )}\left( x\right) P_{m}^{(\alpha ,\beta )}\left(
y\right)  &=&(-1)^{m}\frac{\Gamma (\alpha +m+1)\Gamma (\beta +m+1)}{n!}%
\sum\limits_{k=0}^{m}\frac{(-1)^{k}(\alpha +\beta +m+1)_{k}}{4^{k}(m-k)!} \\
&&\times \sum\limits_{l=0}^{k}\frac{(1+xy-x-y)^{l}(1+xy+x+y)^{k-l}}{%
l!(k-l)!\Gamma (\alpha +l+1)\Gamma (\beta +k-l+1)}.
\end{eqnarray}%
Next, writing 
\begin{equation}
\Gamma (-\alpha-n)=(-1)^{n}\frac{\Gamma (-\alpha)\Gamma (1+\alpha)}{\Gamma (\alpha+n+1)},\quad n \in \mathbb{Z}_+, \alpha + n \notin \mathbb{Z}_{-},
\end{equation}%
and similarly for $\beta$, Bateman's formula takes the following form: 
\begin{multline*}
\Gamma (-\alpha -m)\Gamma (-\beta -m)P_{m}^{(\alpha ,\beta )}\left( x\right)P_{m}^{(\alpha ,\beta )}\left( y\right)  = \frac{(-1)^{m}}{m!}%
\sum_{k=0}^{m}\frac{(\alpha +\beta +m+1)_{k}}{4^{k}(m-k)!}
\\ \sum\limits_{l=0}^{k}\frac{(1+xy-x-y)^{l}(1+xy+x+y)^{k-l}}{l!(k-l)!}%
\Gamma (-\alpha -l)\Gamma (-\beta +l-k).
\end{multline*}
Equivalently, changing the summation order and performing the index change $k \mapsto k+l$, we obtain: 
\begin{multline*}
\Gamma (-\alpha -m)\Gamma (-\beta -m)P_{m}^{(\alpha ,\beta )}\left( x\right)P_{m}^{(\alpha ,\beta )}\left( y\right)  =\frac{(-1)^{m}}{m!}%
\sum_{l=0}^{m}\sum\limits_{k=0}^{m-l}\frac{(\alpha +\beta +m+1)_{k+l}}{4^{k+l}(m-k-l)!} \\
\\ \frac{[(1-x)(1-y)]^{l}[(1+x)(1+y)]^{k}}{l!k!}\Gamma (-\alpha-l)\Gamma (-\beta -k).
\end{multline*}%
Specializing this form of Bateman's formula to $\alpha =-B+\mu _{j}$, $\beta =-B-\mu _{j}$, $x=iX,$ and $\ y=iY,$ we get:  
\begin{multline*}
|\Gamma (B-m+\mu _{j})|^{2}P_{m}^{(-B+\mu _{j},-B-\mu _{j})}\left(iX\right) P_{m}^{(-B+\mu _{j},-B-\mu _{j})}\left( iY\right)  
=\frac{(-1)^{m}}{m!}\sum_{l=0}^{m}\sum_{k=0}^{m-l}\frac{\left(1-2B+m\right) _{k+l}}{(m-k-l)!} \\ \frac{\overline{V}^l V^{k}}{k!l!}\Gamma \left( B-l-\mu_{j}\right) \Gamma \left( B-k+\mu _{j}\right).
\end{multline*}
Keeping in mind \eqref{B6}, Theorem \ref{RepKer} is proved.
\end{proof}

\begin{center}
\textbf{Appendix C}
\end{center}

\begin{proof}[Proof of Proposition \protect\ref{Theta}]
Recall $\alpha(j,B) = 2(j+B)\log(R)/\pi$ and consider the series \eqref{Not1}: 
\begin{equation*}
\sigma _{k,l}^{R,B}(z,w)=\sum_{j\in \mathbb{Z}}\Gamma \left( B-k+i\frac{\alpha(j,B)}{2}\right) \Gamma \left(B-l-i\frac{\alpha(j,B)}{2}\right) \left( \frac{z\overline{w}}{R}\right) ^{j}, 
\end{equation*}%
where without loss of generality, we assume $0 \leq k \leq l \leq m$. Now, perform there the index change $j\mapsto j-B$ to write it as
\begin{align*}
\sigma _{k,l}^{R,B}(z,w) & = \left( \frac{R}{z\overline{w}}\right)^{B}\sum_{j\in \mathbb{Z}}\Gamma \left( B-k+i j \frac{\log (R)}{\pi}\right) \Gamma \left(B-l-i j \frac{\log (R)}{\pi }\right) \left( \frac{z\overline{w}}{R}\right) ^{j} 
\\& = \left( \frac{R}{z\overline{w}}\right)^{B}\sum_{j\in \mathbb{Z}}\prod_{s=1}^{l-k}\left( B-(k+s)+i j \frac{\log (R)}{\pi}\right) \\& \Gamma \left( B-l+i j \frac{\log (R)}{\pi}\right) \Gamma \left(B-l-i j \frac{\log (R)}{\pi }\right) \left( \frac{z\overline{w}}{R}\right) ^{j},
\end{align*}%
where an empty product equals one. Next, we shall appeal to formula 8 from \cite{Erd}, p.4: 
\begin{equation*}
\Gamma \left( B-l+ij\frac{\log R}{\pi }\right) \Gamma \left( B-l-ij\frac{\log R}{\pi }\right) =\frac{2\log (R)[\Gamma (B-l)]^{2}jR^{j}}{R^{2j}-1}\prod_{q=1}^{B-l-1}\left( 1+\frac{(j\log R)^{2}}{\pi ^{2}q^{2}}\right)
\end{equation*}%
to get further 
\begin{multline*}
\sigma _{k,l}^{R,B}(z,w)  = 2\log (R)[\Gamma (B-l)]^{2} \left( \frac{R}{z\overline{w}}\right)^{B}
\sum_{j\in \mathbb{Z}}\prod_{s=1}^{l-k}\left( B-(k+s)+i j \frac{\log (R)}{\pi}\right)
 \\ \prod_{q=1}^{B-l-1}\left( 1+\frac{(j\log R)^{2}}{\pi ^{2}q^{2}}\right)\frac{jR^{j}}{R^{2j}-1}\left( \frac{z\overline{w}}{R}\right) ^{j}.
\end{multline*}%
Now, note that 
\begin{equation*}
\prod_{q=1}^{B-l-1}\left( 1+\frac{(j\log R)^{2}}{\pi ^{2}q^{2}}\right)\frac{jR^{j}}{R^{2j}-1}
\end{equation*}
is invariant under the flip $j\mapsto -j$, while 
\begin{equation*}
 \prod_{s=1}^{l-k}\left( B-(k+s)+i j \frac{\log (R)}{\pi}\right),
\end{equation*}
is a complex polynomial in $j$. The odd part of the latter leads to series of the form: 
\begin{equation}\label{C1}
\sum_{j\geq 1} j^{2s}\prod_{q=1}^{B-l-1}\left( 1+\frac{(j\log R)^{2}}{\pi ^{2}q^{2}}\right) \frac{R^{j}}{R^{2j}-1}\left[ \left( \frac{z\overline{w}}{R}\right) ^{j} - \left( \frac{R}{z\overline{w}}\right) ^{j}\right], \quad s \geq 1
\end{equation}%
and its even part to 
\begin{equation}\label{C2}
\sum_{j\geq 1} j^{2s}\prod_{q=1}^{B-l-1}\left( 1+\frac{(j\log R)^{2}}{\pi ^{2}q^{2}}\right) \frac{j R^{j}}{R^{2j}-1}\left[ \left( \frac{z\overline{w}}{R}\right) ^{j} + \left( \frac{R}{z\overline{w}}\right) ^{j}\right], \quad s \geq 0.
\end{equation}%
But
\begin{equation*}
4i\sum_{j\geq 1}\frac{R^{j}}{R^{2j}-1}\sinh \left( j\log \left( \frac{z%
\overline{w}}{R}\right) \right)
\end{equation*}%
is the logarithmic derivative of the theta function $\theta _{4}$ evaluated
at $(i/2)\ln (z\overline{w}/R)$ (\cite{Erd1}, p.358), and in turn 
\begin{equation*}
4i\sum_{j\geq 1}\frac{jR^{j}}{R^{2j}-1}\cosh \left( j\log \left( \frac{z%
\overline{w}}{R}\right) \right)
\end{equation*}%
is its second logarithmic derivative. Since 
\begin{equation*}
j^{2s}\prod_{q=1}^{B-l-1}\left( 1+\frac{(j\log R)^{2}}{\pi ^{2}q^{2}}\right), \quad s \geq 0,
\end{equation*}
are even polynomials in $j$, then the series \eqref{C1} and \eqref{C2} are higher logarithmic derivatives of $\theta_4$ as well evaluated at $(i/2)\ln (z\overline{w}/R)$. Keeping in mind the expression of the poly-analytic Bergman kernel proved in Theorem \ref{RepKer}, proposition \ref{Theta} is proved.
\end{proof}

\end{document}